\documentclass{article}
\usepackage{amsmath,amsfonts,graphicx,hyperref,color,caption,subcaption,float,soul,multirow,float, fullpage, authblk}
\usepackage{amsthm}
\usepackage{cleveref}
\usepackage{microtype,amssymb}
\usepackage{framed}
\usepackage[symbol]{footmisc}

\usepackage[inline]{enumitem}
\usepackage{bbm}
\usepackage{todonotes}
\usepackage{algorithm,algorithmic}
\usepackage{thm-restate}
\theoremstyle{plain}
\newtheorem{theorem}{Theorem}[section]
\newtheorem{lemma}[theorem]{Lemma}

\newtheorem{corollary}[theorem]{Corollary}
\newtheorem{definition}[theorem]{Definition}

\newtheorem{question}[theorem]{Question}

\def\RR{{\mathbb R}}

\def\EE{{\mathbb E}}
\def\NN{{\mathbb N}}
\def\SS{{\mathbb S}}

\def\cT{{\mathcal T}}

\def\cD{{\mathcal D}}

\def\cX{{\mathcal{X}}}
\def\cR{{\mathcal{R}}}
\def\cF{{\mathcal{F}}}

\def\Pr{{\mathrm Pr}}
\def \d{{\mathrm d}}

\newcommand{\eps}{\varepsilon}

\title{A Query-Driven Approach to Space-Efficient Range Searching}

\author[*,1,2]{Dimitris Fotakis}
\author[*,1]{Andreas Kalavas}
\author[*,1]{Ioannis Psarros}

\affil[1]{Archimedes, Athena Research Center, Greece}
\affil[2]{School of Electrical \& Computer Engineering, National Technical University of Athens, Greece}

\begin{document}
\maketitle
\footnotetext[1]{Partially supported by project MIS 5154714 of the National Recovery and Resilience Plan Greece 2.0 funded by the European Union under the NextGenerationEU Program.}
\begin{abstract}
We initiate a study of a query-driven approach to designing partition trees for range-searching problems.  Our model assumes that a data structure is to be built for an unknown query distribution that we can access through a sampling oracle, and must be selected 
such that it optimizes a meaningful performance parameter \emph{on expectation}. 
Our first contribution is to show that a near-linear sample of queries allows the construction of a partition tree with a near-optimal expected number of nodes visited during querying. We enhance this approach by treating node processing as a classification problem, leveraging fast classifiers like shallow neural networks to obtain experimentally efficient query times. 
Our second contribution is to develop partition trees using sparse geometric separators. Our preprocessing algorithm, based on a sample of queries, builds a balanced tree with nodes associated with separators that minimize query stabs on expectation; this yields both fast processing of each node and a small number of visited nodes, significantly reducing query time.
\end{abstract}

\section{Introduction}
Range searching is a fundamental problem in data structures, with applications in databases and computational geometry. A typical range searching problem 
asks to preprocess a dataset $P$ into a data structure so that for a query range $q$, the set $q\cap P$ can be reported efficiently. While there are various solutions for different families of ranges, space-efficient solutions, i.e., those requiring near-linear space for storing all required information, typically fall into the category of partition trees. Partition trees are a very natural class of data structures characterized by their hierarchical design: each node is associated with a subset $S$ of the input and the subsets associated with its children realize a partition of $S$. Classic examples of partition trees are kd-trees, quadtrees, and interval trees. A survey on data structures for range searching can be found in \cite{Ag04}. Even for simple ranges, such as halfspaces or balls, linear-size data structures suffer from query times of the form $n^{1-O(1/d)}$, where $d$ is typically the ambient dimension (or the VC dimension as in \cite{CW89}). This raises the question of whether one can get improved running times for data structures that adapt to the distribution of the input. 

Designing data-driven algorithms is an important aspect of modern data science and algorithm design. Instead of relying on worst-case guarantees, it is natural to optimize over a class of parameterized algorithms for the given (unknown) distribution of inputs. Optimization is performed with respect to some performance parameter, on expectation over the distribution of inputs, and typically relies on accessing a hopefully small number of samples from the distribution. More formally, given a class of algorithms $\mathcal{A}$ and a utility function $u:\RR^k \times \cX \to \NN$ that characterizes the performance of an algorithm $A_t\in \mathcal{A}$ with input $x\in \cX$, we aim for $t$ that (approximately) minimizes  $\EE_{{x\sim \cD}} [u(t,x)]$, where $\cD$ is the input distribution. An overview of results obtained by this approach can be found in \cite{B20}.  

Towards extending this point of view to data structures, we initiate a study on how well partition trees can adjust to the distribution of the queries. We assume that we are given access to the query distribution through a sampling oracle, and we are interested in bounding the number of required disclosed queries in order to argue about the expected performance of the constructed tree. This setting is motivated by various real-world situations where access to query sets can be ensured before building the data structure. For example, the input dataset may follow the same distribution as the queries; in such a scenario the input dataset can be used to extract information about the queries and build a data structure with improved performance. Moreover, such a ``query-driven" data structure can be used in parallel with other ``query-agnostic" data structures to improve overall performance. For instance, a system that processes a long stream of queries can initially use a data structure with optimal worst-query performance, and build a ``query-driven" data structure only after a certain amount of queries have appeared to improve response time in later queries.

\subsection{Related work}
Our approach is motivated by the fact that the worst-case analysis of algorithms is too pessimistic to give an actual hint about practical algorithm design, and   
can be seen as an attempt to formalize the fact that algorithms can learn to improve their expected performance. For an overview of different approaches to the analysis of algorithms beyond worst-case, the reader is referred to \cite{R21}. The approach more relevant to ours is that introduced in \cite{GR17}, which models algorithm selection as a statistical learning problem. Other notable examples include 
self-improving algorithms~\cite{ACCLMS11}, instance-optimality~\cite{ABC17} and smoothed analysis~\cite{ST09}. Related results on data structures include randomized partition trees whose query running times adapt to the difficulty of the point configuration~~\cite{DS15}, LSH trees whose cutting rules are optimized with respect to the given dataset~\cite{AB22}, and learning-augmented binary trees with optimized search assuming access to advice from a frequency estimation oracle. 
\subsection{Our results}
    We study query-driven data structures for range searching motivated by the following (informally-posed) question. 
    \begin{question}
    Can we build a partition tree with improved query time when we are given a sample of queries?
    \end{question}
    
    Our first contribution is to prove that given a sample of queries whose size is near-linear to the number of input points, we can build a partition tree that has a near-optimal expected \emph{visiting number}, i.e., the number of nodes required to be visited by the query algorithm is (nearly) the smallest possible, on expectation, over the query distribution. To prove this we adapt a known connection between partition and spanning trees, first observed in \cite{CW89} which allows us to reduce our question to that of computing a spanning tree with a small expected stabbing number, i.e., a small expected 
    number of edges whose intersection with a random query contains exactly one vertex. Using known sampling bounds, we bound the number of sampled queries required to estimate the stabbing number on all spanning trees. Then computing a spanning tree with a small stabbing number is essentially a minimum spanning tree (MST) computation where each weight corresponds to the number of queries stabbing each edge. Finally, we rely on \cite{CW89} to transform the resulting MST to a partition tree with a small expected visiting number. Throughout this informal presentation of our results, $\tilde{O}$ hides dependence on $\mathrm{polylog}(n)$, factors depending only on the dimension $d$, or factors depending only on the VC dimension of the query ranges. Our first main result can be summarized as follows.
    \begin{theorem}[Informal version of \Cref{thm:samplingcomplexity}]
        Given a set of $n$ points $P$, and $\tilde{O}(n)$ i.i.d. query ranges sampled from an unknown distribution $\cD_Q$, we can build a partition tree $T$ on $P$ in $\tilde{O}(n^3)$ time, such that $T$ has an expected visiting number within $O(\log n)$ from the optimal for $\cD_Q$. The construction of $T$ succeeds with probability $1-\frac{1}{\mathrm{poly}(n)}$. 
    \end{theorem}

    The visiting number provides an important characterization of the query algorithm's performance, but it provides no information about the time needed to process each node. We address this in practice. In our implementation, we treat the processing of each node, as a classification problem, where one needs to infer whether the query range completely covers the pointset associated with the node (in which case the output contains all points), or whether there is a non-empty intersection with the pointset (in which case the query algorithm visits both children).
    By relying on shallow neural networks (NNs) to support fast inference, we implement a variant of partition trees whose query time is dominated by the visiting number, which is rigorously proven to be nearly optimized with respect to the query distribution. The use of NNs naturally induces errors in the final output, which are experimentally shown to be insignificant. 

    In our second contribution, we focus on range searching where ranges are defined as Euclidean balls of fixed radius. In this setting, the radius is known during preprocessing, and queries are completely described by their centers.  
    We build partition trees that allow for fast (worst-case) processing of each node at the cost of potentially having a large visiting number; an issue that we mitigate in a data-driven manner. Using a sufficiently large sample of query centers, at each level of the partition tree, we greedily choose a ``ring" separator, i.e., a sphere that is as sparse as possible with respect to the query distribution near its boundary and divides the pointset into two balanced parts. 
    This leads to a balanced partition tree where each node is associated with a geometric separator, and the query algorithm only proceeds to both children if the query stabs the separator, meaning that the query has a non-empty intersection with both sides, an event that is fast to determine and unlikely to happen by construction. We also upper bound the query time for a certain class of distributions that generalize the uniform distribution, and we test the efficacy of our method experimentally.  
    Our second main result is the following.     \begin{theorem}[Informal version of \Cref{T424}]
    Given a set of $n$ points $P$ in $\RR^d$, a radius $r>0$,  $\tilde{O}(n)$ query centers sampled from an unknown distribution $\cD_q$, for which the probability that any Euclidean ball of radius $r$ contains a point is at most $\alpha$, we can build a partition tree $T$ on $P$ for the range searching problem of balls of radius $r$, with expected query time in $\tilde{O}(n\alpha^{O(1/d)})$, space complexity in $\tilde{O}(n)$, and preprocessing time in $\tilde{O}(n^{d+2})$. The construction of $T$ is successful with probability $1-\frac{1}{\mathrm{poly}(n)}$.
    \end{theorem}
    Notice that the query time depends on the parameter $\alpha\leq 1$ that captures the sparsity of the query distribution.

    In \Cref{section:prel} we introduce notation and define necessary notions, in \Cref{sec3} we present our results on selecting a partition tree that nearly minimizes the number of visited nodes during query, in \Cref{sec4} we present and analyze separator trees, and in \Cref{section:experiments} we present our experiments.

\section{Preliminaries}
\label{section:prel}
For any set $X$, $2^X = \{x \subseteq X\}$ and ${{X}\choose{2}} = \{(a,b)\mid a\neq b\in X \}$. 
Given a metric $\mathcal{M}= (M,d_M)$, a point $p\in M$ and a radius $r>0$, $B_M(p,r)$ denotes the ball of radius $r$ centered at $p$, i.e., $B_M(p,r) := \{x\in M \mid d_M(x,p)\leq r\}$. Moreover, for any point $p\in M$, and parameters $r_2>r_1>0$, let $R_M(p,r_1,r_2):=B_M(p,r_2)\setminus B_M(p,r_1)$ denote the ring of radii $r_1,r_2$ centered at $p$. We may refer to such a ring as annulus when $\mathcal{M}$ is the Euclidean metric $\ell_2^d$. The doubling constant of a metric space is the smallest integer $c$ such that for any $r>0$ any ball of radius $r$ can be covered by $c$ balls of radius $r/2$. 

A  partition tree $T$ built on a set of points $P$ is a tree where each node $v$ is associated with a subset $P_v \subseteq P$, and for each node $v$ with children $u_1,\ldots,u_{\ell}$, we have $P_v=\bigcup_i P_{u_i}$ and $\forall i\neq j~P_{u_i} \cap P_{u_j} = \emptyset$. We say that a query range $q$ stabs a set $S$ if $S \cap q \neq \emptyset$ and $S \cap q \neq S$. 
For a partition tree $T$, and a query range $q$, the visiting number of $T$ w.r.t.~$q$, denoted by $\zeta(T,q)$ is the number of nodes $v$ in $T$ 
which are either root of $T$ or their parent $u$ is stabbed by $q$ (meaning that $P_u$ is stabbed by $q$).

Fix a countably infinite domain $X$.  the so-called Vapnik-Chernovenkis dimension~\cite{Sau72, She72, VC71} (VC dimension) of a set $\cF$  
of functions from $X$ to $\{0, 1\}$, denoted by $VCdim(\cF)$, is
the largest $d$ such that there is a set $S=\{x_1, \ldots,x_d\}\subseteq X$, for which, for any $s\in 2^S$ there is a function $f\in \cF$ such that $p\in s \iff f(p)=1$. 
A \emph{range space} is defined as a 
pair $(X,\cF)$ where $X$ is the domain and $\cF$ is a set of functions from $X$ to $\{0, 1\}$, or equivalently a 
pair of sets $(X,\cR)$, where $X$ is the domain and $\cR$ is a set of ranges, i.e., subsets of $X$. 
For any $\vec{x} =(x_1,\ldots,x_m) \in X^m$ let $\mu_{f}(\vec{x})=m^{-1}\cdot \sum_{i=1}^m f(x_i)$. For any $\eta>0, r\geq 0,s\geq 0$, let  $d_{\eta}(r,s) = \frac{|r-s|}{\eta+r+s}$. 

\begin{theorem}[{\cite{10.5555/338219.338267}}]
\label{theo:etaepsilonapproxfunctions}
    Let $X$ be a countably infinite domain. Let $\cF$ be a set of functions from $X$ to $\{0,1\}$ with VC dimension $D$ and let $\cD$ be a probability distribution over $X$. Also, let $\epsilon, \delta, \eta \in (0,1)$. There is a universal constant $c$ such that if $m \geq \frac{c}{\eta \cdot \epsilon^2} \cdot \left( D \log \left(\frac{1}{\eta}\right) + \log\left(\frac{1}{\delta}\right) \right)$
    then if we sample $\vec{x} = \{x_1,\ldots,x_m\}$  independently at
random according to $\cD$, then with probability at least $1 - \delta$, for any $f\in \cF$, 
    $d_{\eta}\left( \mu_f(\vec{x}), \EE_{p\sim\cD}[f(p)] \right)\leq \eps.$ 
\end{theorem}

By setting $\eps$ appropriately, we obtain the following corollary. A similar result is obtained in \cite{HS11}. Due to minor differences in the statement, we include a proof (diverted to \Cref{app:corapproxfunctions}).
\begin{restatable}{corollary}{corapproxfunctions}
\label{theo:etaepsilonapproxfunctionsrelative}
    Let $X$ be a countably infinite domain. Let $\cF$ be a set of functions from $X$ to $\{0,1\}$ with VC dimension $D$ and let $\cD$ be a probability distribution over $X$. Also, let $ \delta, \eta \in (0,1)$. There is a universal constant $c$ such that if $m \geq \frac{c}{\eta } \cdot \left( D \log \left(\frac{1}{\eta}\right) + \log\left(\frac{1}{\delta}\right) \right)$
    then if we sample $\vec{x} = \{x_1,\ldots,x_m\}$  independently at
random according to $\cD$, then with probability at least $1 - \delta$, for any $f\in \cF$, 
\[
 \left\lvert   \mu_f(\vec{x})-\EE_{p\sim\cD}[f(p)]\right\rvert \leq (3/8) \cdot \max\{\EE_{p\sim\cD}[f(p)] , \eta\}.
\]
\end{restatable}

\section{Nearly-Optimal Expected Visiting Number} \label{sec3}
In this section, we study the problem of selecting a partition tree that minimizes the expected number of visited nodes during query. 
We are given a set of points $P\subset \RR^d$, and we want to build a partition tree $T$ on $P$ to answer range searching queries following some unknown distribution $\cD_Q$.  
We can see a certain amount of queries $S_Q$ independently chosen from $\cD_Q$ before we build $T$; the goal is to build a tree that is best fitted for $\cD_Q$. Let $\mathcal{T}$ be the set of all possible partition trees built on $P$.  We aim for a tree $\tilde{T}$ such that
\[
\EE_{q\sim \cD_Q} [\zeta(\tilde{T},q)] \leq C\cdot  \min_{T\in\mathcal{T}}\EE_{q\sim \cD_Q} [\zeta(T,q)] , 
\]
where $C$ is as small as possible. 
We are interested in answering two questions regarding $\tilde{T}$, namely what is the minimum required size of $S_Q$ so that we can build $\tilde{T}$, and what is the computational complexity of computing such a tree $\tilde{T}$. We first relate the visiting number of partition trees with the more convenient notion of the stabbing number of graphs.

Let $(X, \cR)$ be a range space. For any graph $G$  with vertices $P\subseteq X$ and any $q\in \cR$, $\sigma(G,q)$ denotes the number of edges of $G$ stabbed by $q$ (an edge is a set of two points). 
The following lemma shows a close relation between partition trees with small visiting numbers and spanning trees with small stabbing numbers. This relation was first established by \cite{CW89}; where they focused on the maximum (over ranges) visiting and stabbing numbers. Our statement is stronger than theirs to fit our needs, but the proof is exactly the same. We include the proof for completeness in \Cref{app:visitingstabbing}.

\begin{restatable}[Adapted from \cite{CW89}]{lemma}{lemmavisitingstabbing}
\label{lemma:visitingstabbing}
Let $( X, \cR)$ be a range space and let $P\subseteq X$ be a set of $n$ points:
\begin{enumerate}
    \item If $T$ is a spanning tree of $P$, then there exists a spanning path $\Pi$ 
    such that for any $q\in \cR$, $\sigma(\Pi,q)\leq 2\cdot \sigma(T,q)$. \label{item1}
    \item If $\Pi$ is a spanning path of $P$, then there exists a balanced binary partition
tree $T$ for $P$ such that for any $q\in \cR$, $\zeta(T,q)\leq (2\lceil \log n \rceil+1) \cdot \sigma(\Pi,q)$. \label{item2} 
\item If $T$ is a partition tree for $P$, then there exists a spanning path $\Pi$ for $P$ such that for any $q\in \cR$, $\sigma(\Pi,q)\leq \zeta(T,q)$. \label{item3}
\end{enumerate}     
\end{restatable}
The following corollary follows from \Cref{lemma:visitingstabbing} (\Cref{item3}).
\begin{corollary}
\label{cor:ubstabbing}
Let $\cD_Q$ be a distribution over query ranges and let $P$ be a set of $n$ points. Let $\Phi$ be the spanning tree of $P$ that minimizes $\EE_{q\sim \cD_Q}[\sigma(\Phi,q)]$ over all spanning trees of $P$. Then, 
\[
\EE_{q\sim \cD_Q}[\sigma(\Phi,q)]
\leq \min_{T\in \cT} \EE_{q\sim \cD_Q}[\zeta(T,q)] . 
\]
\end{corollary}
The following corollary follows directly from \Cref{lemma:visitingstabbing}, \Cref{item1,item2} and their proof.
\begin{corollary}
\label{cor:ubvisiting}
Let $\cD_Q$ be a distribution over query ranges and let $P$ be a set of $n$ points. Let $\Phi$ be a spanning tree of $P$. We can construct a binary partition tree $T$ of $P$ such that  
\[
\EE_{q\sim \cD_Q}[\zeta(T,q)] \leq (4\lceil \log n\rceil +2) \cdot 
\EE_{q\sim \cD_Q}[\sigma(\Phi,q)],
\]
in $O(n)$ running time. 
\end{corollary}

 We apply \Cref{theo:etaepsilonapproxfunctionsrelative} with $\eta = 1/n$ to get an estimate, for each pair of points $(a,b) \in P\times P$ of the probability that $\{a,b\}$ is stabbed by a randomly selected query $q$. Light pairs that are stabbed with probability at most $1/n$ will be estimated within an additive error of $O(1/n)$ while the rest of them will be estimated with a constant multiplicative error. Since the expected stabbing number of any spanning tree is the sum of individual probabilities (one for each edge for the event it gets stabbed),  we obtain the following lemma. 
The full proof can be found in \Cref{app:spanningtreespreserved}. 
\begin{restatable}{lemma}{lemmaspanningtreespreserved}
\label{lemma:spanningtreespreserved}
Let $P$ be a set of $n$ points and let $\cD_Q$ be a query distribution having as support a countably infinite set $\cR$.  
Let $D$ be the dual VC dimension of the range space $(P,\cR)$.  
It suffices to sample $m \in O(n\cdot (D+\log(1/\delta)))$ queries $S_Q$ independently from $\cD_Q$, so that with probability at least $1-\delta$,   
for any spanning tree $\Phi$ of $P$, $m^{-1}\cdot  \sum_{q\in S_Q} \sigma(\Phi,q)$ is in 
$
 \left[\frac{5}{8}\cdot \EE_{q\sim\cD_Q} \left[\sigma(\Phi,q)\right] - \frac{3}{8} , \frac{11}{8}\cdot \EE_{q\sim\cD_Q} \left[\sigma(\Phi,q)\right] + \frac{3}{8} \right]$.
\end{restatable}
We now prove the main result of this section.
\begin{theorem}
\label{thm:samplingcomplexity}
    Let $P$ be a set of $n$ points and let $\cD_Q$ be a query distribution having as support a countably infinite set $\cR$. 
Let $D$ be the dual VC dimension of the range space $(P,\cR)$.  
It suffices to sample a set $S_Q$ of $m \in O(n\cdot (D+\log(1/\delta)))$ queries independently from $\cD_Q$, so that with probability at least $1-\delta$,   we can compute a tree $\tilde{T} \in \cT$ such that 
\[
\EE_{q\sim \cD_Q}\left[ \zeta(\tilde{T},q)\right] \in O\left(\log n \cdot \min_{T\in \cT} \EE_{q\sim \cD_Q}[ \zeta(T,q)]\right).
\]
Moreover, given $S_Q$, the running time to compute $\tilde{T}$ is in
$O(n^3 (D+\log(1/\delta))$. 
\end{theorem}
\begin{proof}
    By \Cref{lemma:spanningtreespreserved}, with probability at least $1-\delta$,   
for any spanning tree $\Phi$ of $P$, $m^{-1}\cdot  \sum_{q\in S_Q} \sigma(\Phi,q)$ is in 
\[
  \left[\frac{5}{8}\cdot \EE_{q\sim\cD_Q} \left[\sigma(\Phi,q)\right] - \frac{1}{3} , \frac{11}{8}\cdot \EE_{q\sim\cD_Q} \left[\sigma(\Phi,q)\right] + \frac{1}{3} \right].
\]
Now, let $\tilde{\Phi}$ be the spanning tree that minimizes $m^{-1} \cdot \sum_{q\in S_Q} \sigma(\Phi,q)$ over all spanning trees $\Phi$. For any spanning tree $\Phi$ let $E(\Phi)$ be the set of its edges. Notice that for any spanning tree $\Phi$,  
$\sum_{q\in S_Q} \sigma(\Phi,q) = \sum_{q\in S_Q} \sum_{e\in E(\Phi)} \sigma(e,q)=  \sum_{e\in E(\Phi)} \sum_{q\in S_Q} \sigma(e,q)$. Hence, we can compute $\tilde{\Phi}$, by first constructing a complete weighted graph on $P$ where we assign weight $\sum_{q\in S_Q} \sigma(e,q)$ to each edge $e\in {{P} \choose {2}}$, and then computing the minimum spanning tree on that graph. The running time of this step is $O(n^2 m)$. 

Let $\Phi^{\ast}$ be the spanning tree that minimizes $\EE_{q\sim\cD_Q} [\sigma(\Phi,q)] $
over all spanning trees $\Phi$ of $P$. 
By \Cref{lemma:spanningtreespreserved} and \Cref{cor:ubstabbing}, $m^{-1} \sum_{q\in S_Q} \sigma(\tilde{\Phi},q)$ is in 
\[   O\left(\EE_{q\sim\cD_Q} [\sigma(\Phi^{\ast},q)]\right) \subseteq  O\left(\min_{T\in \cT} \EE_{q\sim \cD_Q}[\zeta(T,q)]\right). 
\]
Given $\tilde{\Phi}$, by \Cref{cor:ubvisiting} we can construct a binary partition tree $T$ of $P$ such that  
\begin{align*}
    \EE_{q\sim \cD_Q}[\zeta(T,q)] &\in  O\left(\log n \cdot 
\EE_{q\sim \cD_Q}[\sigma(\tilde{\Phi},q)]\right)\\
&\subseteq O\left(\log n \cdot 
\EE_{q\sim \cD_Q}[\sigma(\Phi^{\ast},q)] \right)\\
&\subseteq O\left(\log n \cdot \min_{T\in \cT} \EE_{q\sim \cD_Q}[\zeta(T,q)] \right)
\end{align*}  
in $O(n)$ running time. Hence, the overall running time is $O(n^2 m ) \subseteq O(n^3 (D+\log(1/\delta))$. 
\end{proof}

We remark that applying a uniform convergence result for real-valued functions (see e.g.~\cite{AB02}) directly on functions $\zeta$ or $\sigma$ yields a quadratic sampling complexity (as opposed to near-linear in \Cref{thm:samplingcomplexity}), since both functions evaluate to $[0,n]$.

\section{Separator Trees}
\label{sec4}
In this section, we focus on range-searching problems where ranges are metric balls of fixed radius. We define the \emph{Fixed-Radius-Ball (FRB) searching problem} as follows. Given a set $P$ of $n$ points in a metric $(M, d_M)$ and a search radius $r>0$, the problem consists of building a data structure, that given any query  $q\in M$, outputs $P\cap B_M(q,r)$.  

As a set of candidate solutions to the FRB problem, we consider separator trees, a special case of partition trees that allow for efficient traversing. 
A separator tree is a binary tree that divides the pointset into two disjoint subsets from the root downwards until the size of them becomes constant. The division is guided by a separator 
that is 
\begin{enumerate*}[label=\roman*)]
\item efficient, i.e., determining the side of the separator where a new point lies is time-efficient, 
\item balanced, i.e., both sides of the split contain a constant fraction of input points, and 
\item sparse with respect to the query distribution, i.e., the density ``near" the separator is low. 
\end{enumerate*}
A similar notion of small, balanced separators is studied by \cite{FM06}, although it works on graphs where separators are defined as sets of nodes. The problem of learning low-density hyperplane separators (although without a balance condition) is studied by \cite{pmlr-v5-ben-david09a}.

We restrict ourselves to ring-separators which is a natural class of efficient separators in the sense described above and support FRB queries. A ring-separator is defined by a metric ball and partitions the points into those lying in the interior of the ball and those lying in the exterior (and on the boundary) of the ball. 
\begin{definition}
    A ring separator $\varrho(p,r_s)$ is a ring $R(p,r_s,r_s+2r)$, where $r$ is the search radius of the FRB searching problem instance. This separator splits the pointset $P$ into 2 subsets, $P\cap B(p,r_s+r)\setminus \partial B(p,r_s+r) $ and $P \cap \overline{B(p,r_s+r)}$. We call these subsets inner and outer respectively. Moreover, $R(\varrho(p,r_s)):= R(p,r_s,r_s+r)$.
\end{definition}

Ring separators are chosen in a query-driven manner by seeking to minimize the density of the query distribution which is $r$-near to it. 
The intuition is that queries that are close to the separator contribute more to the query time because in that case, both parts of the split may contain output points and need to be checked. 


\subsection{Ring Trees in General Metrics}
\label{sec:genericds}
In this section, we design a data structure for the FRB searching problem in some arbitrary metric $(M,d_M)$. Our data structure is constructed as a binary tree as follows: starting with the whole pointset $P$ at the root, divide it in two according to a ring separator $\varrho$. Recursively repeat this on every new node, until the size of all the leaf sets is $O(1)$. We call this structure a \emph{ring tree}. 
The generic building algorithm can be seen in \Cref{btree} and requires a subroutine $findSeparator_M(P, S_Q)$ that depends on the metric and computes a separator with respect to a given pointset $P$ and a sample $S_Q$ of queries i.i.d.~sampled from a distribution $\cD_Q$. Ideally, $findSeparator_M(P, S_Q)$ returns separators that accommodate fast (expected) query search, e.g., they are sparse with respect to $\cD_Q$ so that the probability that both children of a node need to be checked is small. Using the separator, the building algorithm creates two children corresponding to the inner and outer subsets respectively (this is implemented by subroutine $split$ in \Cref{btree}). 

\setcounter{algorithm}{0}

\begin{algorithm}
\caption{Build Tree}\label{btree}
\begin{algorithmic}[1]
    \STATE {\bfseries Input:} \textit{pointset} $P$, \textit{query sample} $S_Q$ 
    \STATE N $\gets$ \textbf{new} Node
    \STATE $\text{N}.P = P$
    \IF {$|\text{N}.P|=O(1)$}
        \STATE \textbf{return} N
    \ENDIF
    \STATE $\text{N}.separator \gets findSeparator_M(\text{N}.P,S_Q)$
    \STATE $P.in, P.out \gets split(\text{N}.P,\text{N}.separator)$     \STATE $\text{N}.lchild \gets $PREPROCESS$(P.in,S_Q)$
    \STATE $\text{N}.rchild \gets $PREPROCESS$(P.out,S_Q)$
    \STATE \textbf{return} N
\end{algorithmic}
\end{algorithm}

For a query point $q \in M$, the query procedure starts from the root and checks where $q$ falls on the ring separator $R(o,r_s,r_s+2r)$. If $d_M(q,o)< r_s$, then it follows the child corresponding to the inner subset, else if $d_M(o,q)> r_s+2r$, then it follows the child corresponding to the outer subset. If the query falls in the ring (meaning $r_s\leq d_M(o,q)\leq r_s+2r$), then it proceeds to both children. Upon arriving at a leaf node, we enumerate all points and return those in $B(q,r)$. The complete pseudocode can be found in \Cref{qtree}.

\begin{algorithm}[H]
\caption{Query Search}\label{qtree}
\begin{algorithmic}[1]
\STATE {\bfseries Input:} \textit{query point} $q$, \textit{starting node} N
\IF {N is a leaf Node}
    \STATE $Sol\gets \{\}$
    \FOR{each $p\in$ N.$P$}
        \IF{$d_M(p,q)\leq r$}
            \STATE $Sol \gets Sol \cup \{p\}$
        \ENDIF
    \ENDFOR
    \STATE \textbf{return} Sol
\ENDIF
\STATE $\varrho(o,r_s) \gets$ N.$separator$
\IF {$d_M(q,o) < r_s$}
    \STATE \textbf{return} QUERY($q$, N$.lchild$)
\ELSIF {$d_M(q,o)>r_s+2r$}
    \STATE \textbf{return} QUERY($q$, N$.rchild$)
\ELSE
\STATE  \textbf{return} QUERY($q$,N$.lchild$)$\cup$QUERY($q$,N$.rchild$))
\ENDIF
\end{algorithmic}
\end{algorithm}
It is easy to verify that \Cref{qtree} will always return a correct answer. In fact, given a node with separator $\rho = R(o,r_s,r_s+2r)$, $B(q,r)$ may contain points from both children's pointsets, only if $q\in \rho$, and in such cases, the query procedure searches both children for points.
\subsection{Sparse Distributions in $\ell_2$}
In this section, we restrict ourselves to the Euclidean metric, and we 
study family of distributions of the query points $\cD_Q$ for which there exist separators that are both balanced with respect to the input pointset and sparse with respect to $\cD_Q$. The existence of such separators implies the existence of ring trees with bounded query time.
\begin{definition}
        Let $f$ be the probability density function (PDF) of a distribution $\cD_Q$ over $\RR^d$. We say that $\cD_Q$ is $(\alpha, r)$-sparse if for any Euclidean ball $B$ of radius $r$, it holds $\int_{B} f(x) ~\d x \leq \alpha$.   
\end{definition}
Intuitively, sparse distributions generalize the uniform distribution. 
We show that for any set of points $P$ and any sparse query distribution $\cD_Q$, there is a ring separator that is both balanced with respect to $P$ and sparse with respect to $\cD_Q$.



\begin{lemma}
\label{lemma:separatorexistence}
     Let  $\mathcal{D}_Q$ be an $(\alpha, r)$-sparse distribution and let $f$ be its PDF. In addition, let $P$ be a pointset in $\RR^d$, with $|P|=n$. Lastly, let $c$ be the doubling constant of $\RR^d$. There is a point $p\in\RR^d$ and a radius $r_s$ defining a ring separator $\varrho(p,r_s)$ such that for the ring $R=\{x\in\RR^d \mid \|x-p\|_2 \in [r_s,r_s+2r] \}$ is true that:
\begin{enumerate}
    \item $|\{x\in P \mid \|x-p\|\leq r_s\}|\geq \frac{n}{c+1}$ and $|\{x\in P \mid \|x-p\|\geq r_s+2r\}|\geq \frac{n}{c+1}$,
    \item  $\int_{R} f(x) ~\d x \leq 2{\alpha}^{O(1/d)}$.
\end{enumerate}
\end{lemma}   

\begin{proof}
     Let ${B}_1$ be the Euclidean ball of minimum radius such that $|{B}_1\cap  P| = \frac{n}{c+1}$ and let ${B}_2$ be the concentric ball of twice the radius. Separator $\varrho$ will be somewhere in between. By the definition of ${B}_1$ and the doubling constant we have $|P\setminus {B}_2| \geq  \frac{n}{c+1}$. This is because ${B}_1$ (being the ball with the minimum radius containing $\frac{n}{c+1}$ points) is also the most dense out of the balls with the same radius. Thus, ${B}_2$, which is covered by at most $c$ balls of the same radius as $B_1$,  will contain at most $\frac{cn}{c+1}$ points.
     
     For point 2, we first scale the space so that the radius of the ball ${B}_1$ is equal to 1 (and the radius of ${B}_2$ is equal to 2). Due to the scaling, $r$ changes to $r'$. The separator will be concentric to the balls and will have an inner radius so that it is entirely in ${B}_2 \setminus {B}_1$. Thus, we can assume that $\int_{{B}_2 \setminus {B}_1} f(x) ~\d x > \alpha^{O(1/\log c)}$, otherwise the statement is trivially true. By the definition of the doubling constant, we can cover ${B}_2$ with $c^{\Theta(\log(1/r'))}$ balls of radius $r'$. Hence, 
     \[
        \alpha^{O(1/\log c)} \leq \int_{{B}_2\setminus {B}_1} f(x) ~\d x\leq c^{\Theta(\log(1/r'))} \alpha ,
    \]
    which implies
    $r' \leq  {\alpha }^{\Theta(1/\log c)}$.
     This equation implies that there are at least $\frac{{\alpha}^{-\Theta(1/d)}}{2}$ annulii of width $2r'$ in ${B}_2 \setminus {B}_1$, having overall density at most $1$. Hence, one of those annulii has density at most $2{\alpha}^{O(1/d)}$.
\end{proof}

The proof is inspired by a proof for ball separators in $k$-ply ball systems obtained in \cite{sepH11}. The argument about the balance is the same. Here, instead of the $k$-ply property, we have the $(\alpha,r)$-sparsity property. 

\subsection{Computing a ring separator in $\ell_2$}
\label{ssec:ringcomputation}
In this section, we focus on the problem of computing ring separators in Euclidean metrics, a solution to which can be plugged into the algorithm of \Cref{sec:genericds} (as subroutine $findSeparator_{\ell_2}(P,S_Q)$). 
While \Cref{lemma:separatorexistence} guarantees the existence of good separators for sparse query distributions $\cD_Q$, to actually compute a ring separator, we need a set of queries $S_Q$ independently sampled from $\cD_Q$. Since the VC dimension of the space of rings is bounded, the sample size required to get a good estimation of the density in the ring is also bounded, as formalized in the following lemma.
\begin{lemma}\label{lemma:ringsamplesize}
   Let $Q\subset \RR^d$ be a countably infinite set of points and let $\cD_Q$ be any distribution on them. 
        Also, let $ \delta, \eta \in (0,1)$. There is an $m_0 \in O\left(\frac{1}{\eta } \cdot \left( d \log \left(\frac{1}{\eta}\right) + \log\left(\frac{1}{\delta}\right) \right)\right)$ such that for any $m\geq m_0$, if we sample $x_1,\ldots,x_m$ from $Q$ independently at random, with probability at least $1-\delta$, for any annulus $R$, $\left\lvert  m^{-1}{\sum_{i=1}^m \mathbbm{1}\{x_i  \in R\}}
       -\Pr_{q\sim\cD_Q}[q \in R]\right\rvert $ is at most
        $ (3/8) \cdot\max\{\Pr_{q\sim\cD_Q}[q \in R] , \eta\}$.
        
\end{lemma}
\begin{proof}
    We invoke \Cref{theo:etaepsilonapproxfunctionsrelative} for the set of annulus membership predicates. Each ring is a set-theoretic difference of two Euclidean balls. Hence, its VC dimension is $O(d)$. 
\end{proof}

Given a set of points $P\subset \RR^d$ and a sufficiently large query sample $S_Q$, we can now compute a ring separator. We present a brute-force algorithm that enumerates all combinatorially distinct separators with respect to $P, S_Q$ (a more efficient heuristic is discussed in \Cref{section:experiments}). Each separator $\varrho(x,r_s)$ is uniquely defined by its center $x$ and inner radius $r_s$; hence the space of valid solutions $\SS_S$ is equivalent to  $\RR^d\times\RR_{+}$. We now define a set of polynomials whose zero sets partition $\SS_S$ into open connected cells, such that each cell represents a set of combinatorially equivalent separators, i.e., all separators defined by vectors in the same cell have the same points of $P\cup S_Q$ in their inner, middle, and outer parts. Let $P\cup S_Q=\{p_1,p_2,\ldots,\}$.
For every $i \in [|P\cup S_Q|]$, we define:
    \begin{align*}
        h_{in}^i(x,r_s) &:= \|p_i-x\|_2^2-r_s^2,\\
        h_{mid}^i(x,r_s) &:= \|p_i-x\|_2^2-(r_s+r)^2,\\
        h_{out}^i(x,r_s) &:=\|p_i-x\|_2^2-(r_s+2r)^2.
    \end{align*}
The main observation is that for all points in a cell, the sign of each of the above polynomials remains the same; hence all $x,r_s$ in the same cell define separators realizing the same $P\cup S_Q \cap B(x,r_s)$, $P\cup S_Q \cap B(x,r_s+r)$,  $P\cup S_Q \cap B(x,r_s+2r)$. 
Using known techniques to sample arbitrary points from each cell of the arrangement defined by the above hypersurfaces, we obtain the following theorem.
\begin{theorem}\label{theorem:computingringseparator}
There is an algorithm that given as input a set $P$ of $n$ points in $\RR^d$, and a sample $S_Q$ of $O\left(\frac{1}{\eta } \cdot \left( d \log \left(\frac{1}{\eta}\right) + \log\left(\frac{1}{\delta}\right) \right)\right)$ points 
independently chosen 
from an $(\alpha,r)$-sparse query distribution $\cD_Q$ satisfying the guarantees of \Cref{lemma:ringsamplesize},  computes a ring separator $\varrho(p,r_s)$ in  $O(n)^{d+2} \cdot 2^{O(d)}$ time that satisfies:
\begin{enumerate}
\item $|\{x\in P \mid \|x-p\|\leq r_s\}|\geq \frac{n}{c+1}$ and $|\{x\in P \mid \|x-p\|\geq r_s+2r\}|\geq \frac{n}{c+1}$,
    \item  $\int_{R(\varrho(p,r_s))} f(x) ~\d x \leq O(1)\cdot{\alpha}^{O(1/d)}+O(\eta)$,
\end{enumerate}
where $f$ is the density of $\cD_Q$. 
\end{theorem}
\begin{proof}
Let $\mathcal{A}$ be the arrangement containing the zero sets of all polynomials $h_{in}^i(x,r_s)$, $h_{mid}^i(x,r_s)$, $h_{out}^i(x,r_s)$. Using the result of \cite{BASU199728}, we can compute a set $\mathcal{C}$ of points, in $O(n)^{d+2} \cdot 2^{O(d)}$
time, that contains at least one point in each cell of
$\mathcal{A}$. By definition of $h_{in}^i(x,r_s)$, $h_{mid}^i(x,r_s)$, $h_{out}^i(x,r_s)$, any two points in the same cell correspond to two combinatorially equivalent separators: both have the same points in their inner parts, the same points in their outer parts, and the same points in the two rings. Moreover,  there are no points of $P\cup S_Q$ on any of the three concentric spheres of radii $r_s, r_s+r, r_s+2r$. We simply return the ring separator $\varrho(x,r_s)$, 
that minimizes $|R(\varrho(x,r_s))\cap S_Q|$ among all separators in $\mathcal{C}$ 
whose both inner and outer parts contain at least $\frac{n}{c+1}$ points from $P$. The claim then follows from \Cref{lemma:separatorexistence,lemma:ringsamplesize}. 
\end{proof}


\subsection{Ring Trees in $\ell_2$}
In this section, we design ring trees for the case of sparse query distributions in $\ell_2$. In this case, we are able to bound the time and space complexities for our data structures. 
\begin{theorem}\label{T424}
    Let $\cD_Q$ be an $(\alpha, r)$-sparse query distribution and let $S_Q$ be a set of $ O\left(n \cdot \left( d \log n + \log\left(\frac{1}{\delta}\right) \right)\right)$ query points independently sampled from $\cD_Q$. With probability $1-\delta$ over the choice of $S_Q$ the following holds. 
    For any set $P$ of $n$ points in $\RR^d$, we can build a ring tree that solves the FRB searching problem in the Euclidean space and comes along with the following guarantees:
    \begin{enumerate}
        \item The expected query time is in $O(n\alpha^{O(1/d)}c\log n)$
        \item The space complexity is in $O(dn)$
        \item The preprocessing time is $O(\log n)\cdot O(n)^{d+2} \cdot 2^{O(d)}$.
    \end{enumerate}
\end{theorem}

\begin{proof}
    By \Cref{theorem:computingringseparator}, with probability at least $1-\delta$ over the choice of $S_Q$, we can compute, for any subset $S$ of $P$, ring separators $\varrho(p,r_s)$ that satisfy:
$|\{x\in S \mid \|x-p\|\leq r_s\}|\geq \frac{|S|}{c+1}$, $|\{x\in S \mid \|x-p\|\geq r_s+2r\}|\geq \frac{|S|}{c+1}$, and 
    \[\Pr_{q\sim \cD_Q}[q\in R(\varrho(p,r_s))] \leq O(1)\cdot{\alpha}^{O(1/d)}+O(1/n).\]

 The running time of each such computation is in $O(|S|)^{d+2} \cdot 2^{O(d)}$. Let $b$ be the number of nodes in the tree that the query follows both of their children. We know that the probability that a query follows both children is equal to the probability that the query falls on the ring separator. This probability is upper bounded by  $O(1)\cdot \alpha^{O(1/d)}+O(1/n)$.
    
    Thus, the expected number of bad nodes is $\mathbb{E}[b]=O(n)\cdot 2\alpha^{O(1/d)} + O(1)$, because we have at most $n$ separators in the tree. Finally, the expected query time is less than the product of the leaves visited (which is equal to the expected number of bad nodes plus $1$) and the height of the tree:    
    \[\mathbb{E}[T_q] \leq (\mathbb{E}[b]+1)\cdot h=(O(n)\cdot 2\alpha^{O(1/d)}+O(1))\cdot h.\]
    We have that $h=\log_{\frac{c+1}{c}}{n}=\frac{\log n}{\log{\frac{c+1}{c}}}=O(c\log n)$. We have $\frac{c+1}{c}$ as the base of the logarithm, as the bigger in size child has at most $\frac{c}{c+1}$ portion of its parent points. 
    This concludes the proof of the first point.
    
    For the second point, we simply observe that each point is stored in only one leaf. Consequently, the required space includes the space needed to store the points and the space needed to store at most $n$ inner vertices of the model along with their separators. Thus the space needed is $O(dn)$.

    Lastly for the third point, during preprocessing we need to find $O(n)$ ring separators, as we have $O(n)$ inner vertices. 
    Each layer of the tree realizes a partition of $P$ into subsets $S_1,\ldots S_{\ell}$, and computing all ring separators in this layer costs $2^{O(d)} \cdot \sum_{i=1}^{\ell} O(|S_i|)^{d+2} =  O(n)^{d+2} \cdot 2^{O(d)}$. 
    Hence, the total time required to find all separators is $O(h) \cdot O(n)^{d+2} \cdot 2^{O(d)}$, where $h=O(c\log n)$ as shown before. 
    Also, after finding each separator, we need to split the pointset on that vertex. Overall, the time needed for all splits is upper bounded by $O(nh) \subseteq O(n\cdot c\log n)$. The total preprocessing time is the sum of the times of those two operations (recall that $c=2^{O(d)})$.
\end{proof}

\section{Experiments} 
\label{section:experiments}
In this section, we discuss our experimental results\footnote{Our code can be found at \url{https://github.com/anonymous-752/Query_Driven_Range_Searching}.}. For the first part, we show that the use of shallow NNs is indeed beneficial as they provide fast inference time, with only a minor effect on accuracy. Next, we draw sample sets $S_Q$ and observe how the expected query times for the respective trees $\tilde{T}$ decrease as the sample set size increases. For the second part, we develop a heuristic for finding a good separator and use it to build trees, again on different-sized sample sets. The experiments were performed on an AMD Ryzen Threadripper 3960X 24-Core Processor CPU @ 4.5 GHz, using the Eigen library \cite{eigenweb} and the $\mathsf{-march}=\mathsf{native}$ and $\mathsf{-O3}$ flags for improved runtime. 

\subsection{Pointset and Query Distribution}

For the pointset, we use the MNIST \cite{lecun2010mnist} dataset, on which we reduce its dimension by applying the Johnson-Lindenstrauss transform \cite{DasguptaG03} and normalize it so it has zero mean and unit variance. We end up with $60000$ points in $15$ dimensions. 
As queries, we use ball ranges. For the partition trees the centers' coordinates are generated from the $N(0,1)$ distribution, and the radii are generated from the folded normal distribution with mean $0$ and variance $16$. For ring trees, we create a query set by randomly generating two queries near each original point, ensuring the queries remain close to the original point set and fix the search radius at $1.5$.

\subsection{Neural Network Accuracy}
In the inner nodes of partition trees, the decision of whether a query range stabs the pointset is made by a trained NN. More precisely, the NN predicts whether the query range contains the whole pointset, stabs it, or does not contain any points from it. To train the model, we generate ranges and their corresponding labels, and feed them only once to the network, until its weights converge. This way, we avoid overfitting. We also note that training a model for the whole pointset ($60000$ points) takes around $1$ minute.

While the model can make mistakes in various cases, the ones where a wrong prediction causes a wrong answer are when a range contains the whole pointset or stabs it, but the NN predicts that the range is empty. This is because in these cases, points that do belong in the range are missed. Thus, we define ``context-aware accuracy" as the proportion of instances in which the NN's predictions are either correct or result in errors that do not lead to a loss of points. In \Cref{tab:nn_acc} we witness the actual accuracy as well as the context-aware accuracy of 2 NNs (having 16-8 and 100-20 neurons in two hidden layers) and different-sized pointsets. We notice that the accuracy does not change significantly for larger networks, thus we stick to the smaller ones so that we have faster inference times, and consequently faster query times.
\begin{table}[ht]
    \centering
    \vskip -0.1in
    \caption{NN Accuracies}
        \vskip 0.1in
    \label{tab:nn_acc}
    \begin{small}
    \begin{tabular}{c|c|c|c|c}
         \hline
         Pointset Size & \textbf{5000}& \textbf{10000}&\textbf{30000}& \textbf{60000}  \\ \hline
         \multirow{2}{*}{Hidden Layers 16-8} & 95.17 & 95.89 & 94.93 & 95.35\\ 
         &98.30 &98.60 &98.23 &98.63\\
         \hline
        \multirow{2}{*}{Hidden Layers 100-20} & 95.63 &95.97  &96.53  &96.02  \\
        & 98.39& 98.77& 98.98& 98.77\\ \hline
    \end{tabular}
    \end{small}
    \vskip -0.2in
\end{table}

\subsection{Heuristic for Finding Separators}

The complexity of the algorithm of \Cref{theorem:computingringseparator} is prohibitive for $d=15$. Therefore, we implement a heuristic variant of our algorithm based on local search to find a good separator. 
While the algorithm of \Cref{theorem:computingringseparator} explores the solution space $\SS_S$ exhaustively by considering points from all cells in $\SS_S$ (see \Cref{ssec:ringcomputation} for relevant definitions), our heuristic searches for a solution among the vertices of the arrangement of hypersurfaces in $\SS_S$.
It starts on a random vertex and moves locally to the best neighboring vertex, i.e., the vertex corresponding to the sparsest separator among those vertices defined by the same hypersurfaces except one. 
The complete pseudocode can be found in \Cref{sec:seplh}. 
As is often the case with heuristics, it does not offer reliable guarantees regarding the optimality of the resulting separator. 
The time complexity however is reduced, compared to \Cref{theorem:computingringseparator} by limiting the number of iterations.

\subsection{Test Results}

For partition trees, we built trees with different-sized query samples. These samples are used to define the order of the points placed in the leaves, as explained in \Cref{sec3}. We then test them on the same query test set. When the answer to a query range contains a lot of points, the query time is apparently longer. In order to accurately assess a tree's query time, we calculate the average query time per point in the answer. In \Cref{tab:bptree_qtime}, we can see the average query time per point of the output.
\begin{figure}[]
\vskip -0.05in
\begin{center}
\centerline{
\begin{subfigure}{0.5\columnwidth} 
        \centering
        \includegraphics[width=\columnwidth]{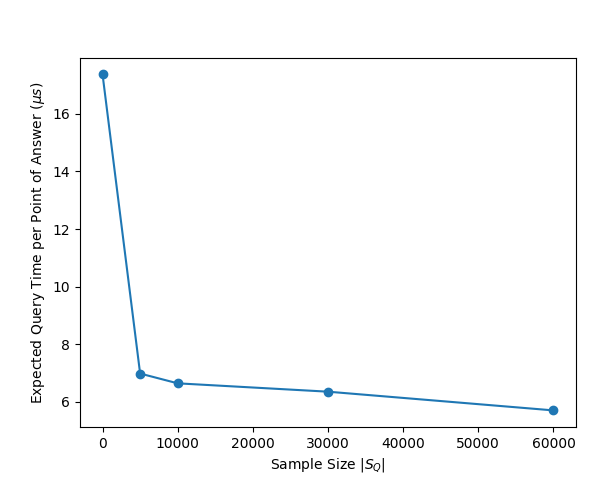}
        \caption{Partition Trees}
        \label{tab:bptree_qtime}
    \end{subfigure}
    \begin{subfigure}{0.5\columnwidth}
        \centering
        \includegraphics[width=\columnwidth]{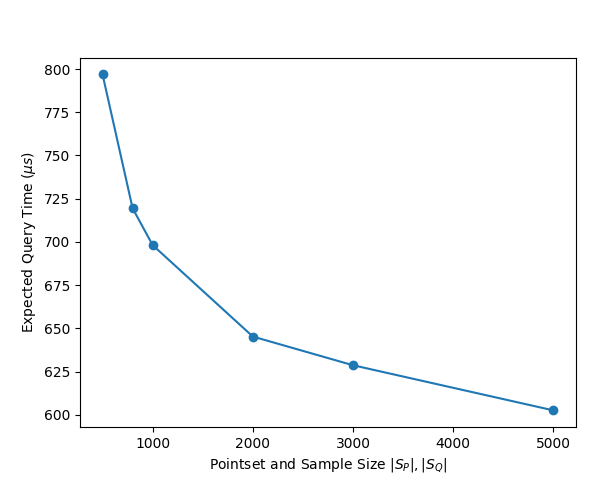}
        \caption{Ring Trees}
        \label{qt_ring}
    \end{subfigure}
}
\vskip -0.05in
\caption{Average query times for partition trees (\ref{tab:bptree_qtime}) and ring trees (\ref{qt_ring}). For partition trees, times are per output point.}
\label{icml-historical}
\end{center}
\vskip -0.1in
\end{figure}
The query times are reduced (and rather fast) as the number of samples used to decide the order of the points in the leaves grows, a fact that is in agreement with the corresponding theoretical results. Additionally, in \Cref{tab:bp_qacc} we see that the structure misses only a few of the points that it should have returned.

\begin{table}
    \centering
    \vskip -0.2in
    \caption{Test Accuracy}
    \label{tab:bp_qacc}
    \vskip 0.1in
    \begin{small}
    \begin{tabular}{c|c|c|c|c|c}
         \hline
         Sample Size $(\cdot 10^3)$& \textbf{0}& \textbf{5}&\textbf{10}& \textbf{30} & \textbf{60}  \\ \hline
         Accuracy & 98.1 & 98.1 & 98.3 & 98.3 & 98.0\\ \hline
    \end{tabular}
    \end{small}
    \vskip -0.2in
\end{table}

For ring trees, we once more build trees on different-sized query samples. However, this time we also sample from the pointset (the pointset sample being the same size as the corresponding query sample), as the preprocessing times rapidly get too long. In \Cref{qt_ring} we can see how the expected query times are reduced. We note that, in contrast with the partition trees, the query times for the ring trees are measured for the queries to their entirety (not per output point), and the accuracy of them is $1$ by definition.

\section*{Impact Statement}

This paper presents work whose goal is to advance the field
of Machine Learning by further exploring its interactions with Theoretical Computer Science. There are many potential societal
consequences of our work, none of which we feel must be
specifically highlighted here.

 \bibliography{jlann}
\bibliographystyle{plain}

 \newpage
\appendix
\onecolumn

\section{Proof of \Cref{theo:etaepsilonapproxfunctionsrelative}}
\label{app:corapproxfunctions}
\corapproxfunctions*
\begin{proof}
 We first consider the case $ \EE_{q\sim \cD} [f(q)] \leq \eta$. 
By \Cref{theo:etaepsilonapproxfunctions}, with $ \eps = 1/9$, we obtain 
\[d_{\eta}\left(\EE_{q\sim \cD}[f(q)],\mu_f(\vec{x})\right)\leq  1/9.\] Hence, 
    $\left|\EE_{q\sim \cD}[f(q)]-\mu_f(\vec{x})\right|\leq  (1/9) \cdot (\eta + \EE_{q\sim \cD}[f(q)] + \mu_f(\vec{x}))\leq \frac{2\eta}{9} + \frac{\mu_f(\vec{x})}{9}$ which implies $\left|\EE_{q\sim \cD}[f(q)]-\mu_f(\vec{x})\right|\leq \frac{\eta}{3}$, since 
    \[
    \EE_{q\sim \cD}[f(q)]\geq \mu_f(\vec{x})
    \implies \left|\EE_{q\sim \cD}[f(q)]-\mu_f(\vec{x})\right|\leq \frac{2\eta}{9} + \frac{\EE_{q\sim \cD}[f(q)]}{9} \leq \frac{\eta}{3},
    \]
    and 
    \begin{align*}
        \EE_{q\sim \cD}[f(q)]< \mu_f(\vec{x})
    &\implies
       \left|\EE_{q\sim \cD}[f(q)]-\mu_f(\vec{x})\right|= \mu_f(\vec{x}) - \EE_{q\sim \cD}[f(q)] \leq \frac{2\eta}{9} + \frac{\mu_f(\vec{x})}{9}\\
    &\implies
        \left|\EE_{q\sim \cD}[f(q)]-\mu_f(\vec{x})\right|=\mu_f(\vec{x}) - \EE_{q\sim \cD}[f(q)] \leq \frac{3\eta}{8}.
    \end{align*}
    Next, we consider the case $\EE_{q\sim \cD} [f(q)] > \eta$. By \Cref{theo:etaepsilonapproxfunctions} with $\eps =1/9$,  $\left|\EE_{q\sim \cD}[f(q)]-\mu_f(\vec{x})\right|\leq  \frac{1}{9} \cdot \left(\eta + \EE_{q\sim \cD}[f(q)] + \mu_f(\vec{x})\right)\leq \frac{ 2\cdot\EE_{q\sim \cD}[f(q)]}{9} + \frac{\mu_f(\vec{x})}{9}$, which implies the claim since:
    \[
    \EE_{q\sim \cD}[f(q)]\geq \mu_f(\vec{x})
    \implies 
    \left|\EE_{q\sim \cD}[f(q)]-\mu_f(\vec{x})\right|\leq  \frac{1}{3}\cdot \EE_{q\sim \cD}[f(q)],
    \]
    and
    \begin{align*}
    \EE_{q\sim \cD}[f(q)]< \mu_f(\vec{x})
    &\implies \left|\EE_{q\sim \cD}[f(q)]-\mu_f(\vec{x})\right|= \mu_f(\vec{x}) - \EE_{q\sim \cD}[f(q)] \leq \frac{3}{8}\cdot \EE_{q\sim \cD}[f(q)].   
    \end{align*}
\end{proof}

\section{Proof of \Cref{lemma:visitingstabbing}}
\label{app:visitingstabbing}
\lemmavisitingstabbing*
\begin{proof}
    \ref{item1}. Let $\Pi$ be the path realized by a depth-first search traversal of $T$. Let $e$ be an edge of $\Pi$ stabbed
by a query range $q$. If $e$ is an edge of $T$ then $e$ contributes the same to $\sigma(\Pi,q)$ and $\sigma(T,q)$. If $e$ is not an edge of $T$, then it creates a cycle in $T$, at least two
of whose edges are stabbed by $q$. The claim then follows by observing that there is no need to charge each edge of $T$ more than twice, due to the depth-first search traversal. 

\ref{item2}. We build a complete binary tree $T$ on $n$ leaves and associate the points of $\Pi$, in sequence, with the leaves of $T$ from left to right. If the pointset associated with an internal node of $T$ is stabbed by a query range $q$, then the subtree rooted at that node must have two consecutive leaves with points $x$ and $x'$ such that
the edge $(x, x')$ of $\Pi$ is stabbed by $q$. By definition there are no more than $\sigma(\Pi,q)$
stabbed edges, hence
no more than $\sigma(\Pi,q) \cdot \lceil\log n \rceil$ stabbed pointsets associated with internal nodes.
The visiting number $\zeta(\Pi,q)$ is at most twice the number of stabbed node-pointsets plus the root.

\ref{item3}. We assign an arbitrary left-to-right order among the children of
every internal node of $T$, and let $x_1,\ldots,x_n$ denote the points in the leaves of $T$ from left to
right. We set $\Pi$ as $x_1,\ldots,x_n$. If the edge $(x_i,x_{i+1})$ is stabbed by a query range $q$, we charge this event to the unique child $u$ of the nearest common ancestor $v$ of $x_i$, and $x_{i+1}$, that is also
an ancestor of $x_i$ (or $x_i$ itself). Since the pointset associated with $v$ is necessarily stabbed by $q$, $u$ must be visited. Furthermore, such a node cannot be charged
twice. 
\end{proof}

\section{Proof of \Cref{lemma:spanningtreespreserved}}
\label{app:spanningtreespreserved}
\lemmaspanningtreespreserved*
\begin{proof}


    We apply \Cref{theo:etaepsilonapproxfunctionsrelative} with $\eta = 1/n$ on the set of functions $\cF$ defined by the pairs of points $P\times P$ and the domain $X$ corresponding to the set of query ranges $\cR$. For each $a,b\in P$, we have a function $f_{ab}\in \cF$ such that for any query range $q$: $f_{ab}(q) = 1 \iff $ $q$ stabs $\{a,b\}$. It remains to bound the VC dimension of the range space $(\cR, \cF)$. Let $\cF' = \{f_a\mid  a\in P\}$ be a set of functions defined such that for any query range $q$: $f_a(q)=1\iff a\in q$. Notice that $f_{ab}(q)=1 \iff (f_a(q)=1 \land f_b(q)=0) \lor (f_a(q)=0 \land f_b(q)=1)$.  
    The dual VC dimension of $(P,\cR)$ is $D$, meaning that the VC dimension of $(\cR,\{f_a \mid a\in P\}$ is $D$, implying that the VC dimension of $(\cR, \cF)$ is $O(D)$ (since each function of $\cF$ is defined as $O(1)$ and/or operations on functions of $\cF'$). Let $\vec{x}=(x_1,\ldots,x_m)$ be as in \Cref{theo:etaepsilonapproxfunctionsrelative}, i.e., each $x_i$ is a query range sampled independently at random from $\cD$, and $m\in O\left(n \left( D \log \left(n\right) + \log\left(\frac{1}{\delta}\right) \right)\right)$. 

    By \Cref{theo:etaepsilonapproxfunctionsrelative}, with probability at least $1-\delta$ over the choice of $\vec{x}$, for all ``light" pairs of points, i.e., pairs that are stabbed with probability at most $1/n$, we can estimate the probability through $\vec{x}$ with an additive error of $3/(8n)$ whereas for all ``heavy'' pairs of points, i.e., those that are stabbed with probability more than $1/n$, we can estimate the probability of getting stabbed with a multiplicative error of $1\pm 3/8$. 
    Now let $\Phi$ be any spanning tree of $P$ and let $E(\Phi)$ be the set of its edges. Let $E_{\ell}\subseteq E(\Phi)$ be the set of light edges of $\Phi$, i.e., edges $(a,b)\in E(\Phi)$ for which the probability that $\{a,b\}$ is stabbed is at most $1/n$ and let $E_{h}\subseteq E(\Phi)$ be the set of heavy edges of $\Phi$, i.e., edges $(a,b)\in E(\Phi)$ for which the probability that $\{a,b\}$ is stabbed is more than $1/n$. Let 
    $\tilde{\sigma}(\Phi,\vec{x})=
        m^{-1}\cdot {\sum_{i=1}^m \sigma(\Phi,x_i)}$. 
    We have,
    \begin{align*}
        \tilde{\sigma}(\Phi,\vec{x})&=
        \frac{1}{m} \cdot \sum_{i=1}^m \sigma(\Phi,x_i)\\
        &= \frac{1}{m} \cdot \sum_{i=1}^m \sum_{e\in E(\Phi)}\sigma(e,x_i)\\ &= 
   \frac{1}{m} \cdot \sum_{i=1}^m \left(\sum_{e\in E_{\ell}}\sigma(e,x_i)+\sum_{e\in E_{h}}\sigma(e,x_i) \right)\\
    & = \sum_{(a,b)\in E_{\ell}} \mu_{f_{ab}}(\vec{x}) +  \sum_{(a,b)\in E_{h}} \mu_{f_{ab}}(\vec{x})\\
    &\geq 
    \frac{5}{8}\cdot \EE_{q} \left[\sigma(\Phi,q)\right] - \frac{3}{8} .
    \end{align*}
   Similarly,
      \begin{align*}
       \tilde{\sigma}(\Phi,\vec{x})&= \frac{1}{m} \cdot \sum_{i=1}^m \sum_{e\in E(\Phi)}\sigma(e,x_i)\\
    & \leq  
    \sum_{(a,b)\in E_{\ell}} 
    \left(\EE_{q}[f_{ab}(q)] +\frac{3}{8n}\right)+
       \sum_{(a,b)\in E_{h}} 
    \frac{11}{8}\cdot \EE_{q}[f_{ab}(q)]\\
    &\leq 
    \frac{11}{8}\cdot \EE_{q\sim\cD_Q} \left[\sigma(\Phi,q)\right]  + \frac{3}{8} .
    \end{align*}
    Hence, for any spanning tree $\Phi$ of $P$, $\tilde{\sigma}(\Phi,\vec{x})$ is in
    \[\left[\frac{5}{8}\cdot \EE_{q\sim\cD_Q} \left[\sigma(\Phi,q)\right] - \frac{3}{8} , \frac{11}{8}\cdot \EE_{q\sim\cD_Q} \left[\sigma(\Phi,q)\right] + \frac{3}{8} \right]. 
    \]
\end{proof}

\section{Heuristic for finding separators}
\label{sec:seplh}
\begin{algorithm}
\caption{Find Separator - Locality Heuristic}\label{seplh}
\begin{algorithmic}[1]
    \STATE {\bfseries Input:} {\textit{pointset} $P$, \textit{query sample} $S_Q$}
    \STATE $min\_mass \gets n$
    \FOR{i in $[m]$}
        \STATE $A \gets$ random subset of $P\cup S_Q$ of size $d+1$
        \WHILE{\textbf{true}}
            \STATE $temp\_mass \gets min\_mass$
            \FOR{each possible ring separator $\varrho$ defined by $neighborhood(A)$}
                \STATE $mass \gets |R(\varrho) \cap S_Q|$
                \IF{$mass < temp\_mass$ \textbf{and} $\varrho$ is $P-balanced$}
                    \STATE $temp\_mass \gets mass$
                    \STATE $t\_ans \gets \varrho$
                \ENDIF
            \ENDFOR
            \IF{$temp\_mass < min\_mass$}
                \STATE $min\_mass \gets temp\_mass$
                \STATE $ans \gets t\_ans$
            \ELSE
                \STATE \textbf{break}
            \ENDIF
        \ENDWHILE
    \ENDFOR
    \STATE \textbf{return} $ans$
\end{algorithmic}
\end{algorithm}

\end{document}